\documentclass[11pt]{article}
\usepackage{chao}
\usepackage{algo}
\usepackage{booktabs,tabularx}

\DeclareMathOperator*{\opt}{OPT}

\title{A Lagrangian Approach to Interdiction of Linear Minimization Problems}
\author{Yu Cong\thanks{University of Electronic Science and Technology of China.
\email{yucong143@gmail.com}.}
\and Kangyi Tian\thanks{University of Electronic Science and Technology of China.
\email{kangyitian947@gmail.com}.}}
\date{}

\begin{document}
\maketitle

\begin{abstract}
We study budgeted interdiction of linear minimization problems.  Given a
family of feasible sets, element weights, and deletion costs, the objective
is to delete elements within a prescribed budget and minimize the remaining
weight of a feasible set.  We give an alternative, simpler Lagrangian
approach that generalizes the connectivity interdiction framework of Huang, Obscura Acosta, and
Yingchareonthawornchai (IPCO' 24).
When the interdiction optimum is positive, we show that
the optimal Lagrange multiplier gives truncated element weights,
in which every feasible set supporting an optimal
interdiction solution has weight less than twice the minimum under the
truncated weights.
We apply the reduction to graph cuts, $k$ cuts, and cocircuits of
perturbed graphic matroids.  We obtain deterministic fully polynomial-time
approximation schemes for graph connectivity interdiction and for cogirth
interdiction on perturbed graphic matroids of fixed perturbation rank.
For graph connectivity, the running time is
$\tilde O(n^3m^2/\varepsilon)$.
The $k$-cut application gives a fully polynomial-time randomized
approximation scheme.
\end{abstract}

\section{Introduction}

The edge connectivity of an undirected multigraph is the minimum
number of edges whose removal disconnects it.  It measures the graph's
resilience to edge failures.  \emph{Connectivity interdiction} asks how far
this resilience can be reduced by deleting edges within a given budget.
In the weighted version, each edge has a weight and a deletion cost, and
the objective is to minimize the weighted edge connectivity of the graph
remaining after the deletions.
Equivalently, we seek a cut whose remaining weight can be made as small
as possible by deleting some of its edges within the budget.  Deletions
outside the chosen cut do not affect its remaining weight, so they can
be omitted.  This cut formulation separates the choice of a cut from
the choice of deletions within it.

This formulation extends naturally to interdiction of linear
minimization problems by replacing the family of cuts with an
arbitrary family of feasible sets.
The underlying problem seeks a feasible set of minimum total weight;
its interdiction version allows elements of the chosen set to be deleted
within a budget, and minimizes the weight that remains.  We call the
chosen feasible set a \emph{witness}.  Connectivity interdiction is the
special case in which witnesses are cuts.  Other examples include
minimizing the residual weight of a $k$-cut and minimizing the residual
weight of a cocircuit in a matroid.  In each case, a witness is a
fixed set of elements whose objective contributions can be removed.

For a fixed witness, choosing which elements to retain is a minimum
knapsack problem.  The main difficulty is to find a useful collection
of witnesses without examining the entire feasible family, which may
be exponentially large.
This leads to the question we study: how can optimization over the
underlying family help identify witnesses for optimal or approximately
optimal interdiction solutions?

\paragraph{Previous work}
Interdiction problems ask how much an optimized system can be disrupted
by removing elements within a budget.  They have been studied for a wide
range of combinatorial optimization problems, with objectives such as
increasing the weight of a minimum spanning tree~\cite{lin_chern_vital_1993}
or the length of a shortest path~\cite{israeli_wood_shortest_2002}, and
decreasing the weight of a maximum matching~\cite{zenklusen_matching_2010}
or the value of a maximum flow~\cite{wollmer_arcs_1964,wood_interdiction_1993}.
Even when the underlying problem is solvable in polynomial time, choosing
the most damaging removals can be NP-hard, as shown for minimum spanning
trees and matchings~\cite{lin_chern_vital_1993,zenklusen_matching_2010}.

Zenklusen~\cite{zenklusen_connectivity_2014} introduced
connectivity interdiction in the above form.  Using a connection to budgeted
minimum cuts, he gave a polynomial-time approximation scheme (PTAS) for
arbitrary deletion costs and polynomial-time exact algorithms for unit
deletion costs.  He also observed that the general problem is NP-hard even
on two vertices joined by parallel edges, where the deletion problem
contains knapsack.  He left open whether the problem with arbitrary
deletion costs admits a fully polynomial-time approximation scheme
(FPTAS).  Huang, Obscura Acosta, and
Yingchareonthawornchai~\cite{huang_fptas_2024,huang_connectivity_2026}
answered this question affirmatively.
Their algorithm introduces an intermediate problem, called
\emph{normalized minimum cut}, which minimizes the remaining weight of
a cut divided by one plus the unused deletion budget.  They use an
approximation to its optimal value to truncate edge weights, enumerate
near-minimum cuts under the truncated weights, and solve a minimum
knapsack problem on each candidate cut.  The role of this normalized
objective in choosing the truncation parameter is not immediate.

For rank-$p$ perturbed graphic matroids, Cong, Xu, and
Zhao~\cite[Theorem~9]{cong_contraction_2026} compute the cogirth
deterministically in $2^{O(p^2)}m^{O(1)}$ time, where $m$ is the number
of elements.  Their packing and cocircuit-enumeration methods supply
the subroutines for the cogirth-interdiction application below.

\paragraph{Our contribution}
We give a simpler explanation of the framework of Huang, Obscura Acosta, and
Yingchareonthawornchai~\cite{huang_connectivity_2026} through linear programming duality.  Dualizing the deletion
budget assigns a price to each unit of deletion cost: retaining an
element contributes its weight, while deleting it contributes its
deletion cost times the multiplier.  Choosing the cheaper option gives
the truncated weight.  Optimality of the Lagrange multiplier then yields
a short proof that every
witness supporting an optimal interdiction solution has truncated weight
strictly less than twice the minimum truncated weight of any feasible set.
The argument applies to
arbitrary feasible families, including cocircuits of general matroids.
Our meta-algorithm combines parametric optimization, enumeration of
strict $2$-approximate minimizers, and minimum knapsack; its running time
is the sum of the costs of these subroutines.

Table~\ref{tab:applications-comparison} compares the resulting algorithms
with the prior guarantees for arbitrary deletion costs.
Throughout, we treat $p$ in rank-$p$ perturbed graphic matroids and $k$
in $k$-cut interdiction as fixed constants, so the resulting schemes are
an FPTAS and an FPRAS, respectively.

\begin{table}[t]
\centering
\small
\setlength{\tabcolsep}{9pt}
\begin{tabularx}{\textwidth}{@{}>{\raggedright\arraybackslash}p{0.25\textwidth}
                              >{\raggedright\arraybackslash}X
                              >{\raggedright\arraybackslash}X@{}}
\toprule
Problem & Previous result & This paper \\
\midrule
Graph connectivity
& FPTAS in
  $O(n^4m^2\ell/\varepsilon)$ \newline
  \cite[Theorem~1]{huang_connectivity_2026}
& FPTAS in
$\tilde O(n^3m^2/\varepsilon)$
  \autoref{thm:connectivity}\\
\addlinespace
$k$-cut
& FPTAS in
  $\tilde O\bigl(\ell T_{\mathrm{LP}}+\ell m^3(nk)^{O(k)}/\varepsilon\bigr)$\newline
  \cite[Section~3.1]{huang_connectivity_2026}
& FPRAS in
  $\tilde O_k(n^{2k}m^2/\varepsilon)$ \autoref{thm:kcut}\\
\addlinespace
Cogirth interdiction\newline for rank-$p$ perturbed graphic matroids
& Compute cogirth in
  $2^{O(p^2)}m^{O(1)}$
  \cite[Theorem~9]{cong_contraction_2026}
& FPTAS in
  $2^{O(p^2)}m^{O(1)}/\varepsilon$
  \autoref{cor:pgm-interdiction} \\
\bottomrule
\end{tabularx}
\caption{Comparison for arbitrary deletion costs.  Interdiction schemes
achieve ratio $1+\varepsilon$ and respect the budget.
In the matroid row, the previous result concerns cogirth
without interdiction.
Here $n$ is the number of vertices, $m$ the number of ground-set elements, and
$\ell=\log(nw_{\max}b)$ with $w_{\max}=\max_e w(e)$ and $b\geq1$.
For the previous $k$-cut result, $T_{\mathrm{LP}}$ is the maximum LP
preprocessing time per weight guess.  Its displayed running time combines
this preprocessing with the candidate count and minimum-knapsack cost
in the cited paper.
}
\label{tab:applications-comparison}
\end{table}

For graph connectivity interdiction, the meta-algorithm gives a
deterministic FPTAS in
$\tilde O(n^3m^2/\varepsilon)$ time; see \autoref{thm:connectivity}.
The recent cut-counting theorem of Xu and Yang~\cite{xu_yang_small_cuts_2026}
bounds the number of cuts strictly below twice the minimum cut value by
$O(n^3)$.  Our strict witness bound allows us to apply this theorem,
and an output-sensitive enumeration algorithm of Yeh, Wang, and
Su~\cite{yeh_enumerating_cuts_2010} gives the stated running time.
Compared with Huang, Obscura Acosta, and Yingchareonthawornchai's
deterministic bound, this reduces the power of $n$ from four to three.
This improvement relies on our strict factor-$2$ witness bound; it does
not follow merely by applying the new cut-counting theorem to their
method.  Their FPTAS requires enumerating
$2.1$-approximate cuts~\cite[Section~3]{huang_connectivity_2026},
so it cannot directly use the cubic bound for strict $2$-approximate cuts.
Parametric search also removes the factor $\ell=\log(nw_{\max}b)$.

For $k\geq3$, combining the reduction with near-minimum $k$-cut
enumeration of Gupta, Harris, Lee, and Li~\cite{gupta_kcut_2022} yields
a fully polynomial-time randomized approximation scheme (FPRAS), using
$\tilde O_k(n^{2k})$ candidates and running in
$\tilde O_k(n^{2k}m^2/\varepsilon)$ time;
\autoref{thm:kcut} states the full bound, including multiplier computation.
The $k$-cut application already appears in
\cite[Section~3.1]{huang_connectivity_2026}; we state the bound obtained
by substituting Gupta, Harris, Lee, and Li's enumeration routine into the same
meta-algorithm.

For cogirth interdiction on rank-$p$ perturbed graphic matroids, we
obtain an FPTAS.  Adapting the two-stage
packing argument of Cong, Xu, and Zhao~\cite{cong_contraction_2026} to
near-minimum cocircuit enumeration and using parametric search gives
running time $2^{O(p^2)}m^{O(1)}/\varepsilon$;
see \autoref{cor:pgm-interdiction}.

These algorithmic bounds are consequences of implementing the shared
truncation framework.

\paragraph{Organization}
Section~\ref{sec:preliminaries} gives the formal model.
Section~\ref{sec:witnesses} proves the Lagrangian witness bounds, and
Section~\ref{sec:algorithms} develops the algorithmic reduction and
multiplier search.  Section~\ref{sec:applications} gives the applications
to graph connectivity, $k$-cuts, and cogirth in perturbed graphic
matroids.

\section{Preliminaries}\label{sec:preliminaries}

Let $E$ be a finite ground set of size $m$, and let
$\mathcal F\subseteq 2^E$ be a nonempty family of feasible sets.
Each element $e\in E$ has a positive integer weight $w(e)$ and deletion
cost $c(e)$.  For $H\subseteq E$, write
$w(H)=\sum_{e\in H}w(e)$ and $c(H)=\sum_{e\in H}c(e)$.
Given an integer budget $b\geq0$, the interdiction problem is
\begin{equation}\label{eq:interdiction}
    \opt=\min\{w(S\setminus R):S\in\mathcal F,\ R\subseteq S,\ c(R)\leq b\}.
\end{equation}
The set $S$ is a \emph{witness}, and $R$ is its deletion set.
The family $\mathcal F$ remains fixed during interdiction; deletion
changes only a witness's objective contribution.
Restricting $R$ to $S$ loses no generality, since deletions outside $S$
do not affect its remaining weight.  For a fixed witness, the retained
set $Q=S\setminus R$ must satisfy $c(Q)\geq c(S)-b$; minimizing $w(Q)$
under this constraint is a minimum knapsack problem.

For connectivity interdiction, $G=(V,E)$ is an undirected multigraph
with $n\geq2$ vertices and $m$ edges.  For $A\subseteq V$, let
$\delta_G(A)$ denote the edges with exactly one endpoint in $A$.
We regard cuts as edge sets and take
$\mathcal F=\{\delta_G(A):\emptyset\neq A\subsetneq V\}$ in
\eqref{eq:interdiction}.  We omit the subscript $G$ when the graph is clear.
We follow Oxley~\cite{oxley_matroid_2011} for matroid terminology.

Throughout, $0<\varepsilon\leq1$, $\tilde O$ suppresses polylogarithmic factors in the
combinatorial input size and $1/\varepsilon$, and $\tilde O_k$ may
additionally hide factors depending on $k$.

\section{Lagrangian witnesses}\label{sec:witnesses}

The deletion budget couples the choices made for individual elements.
Assigning a multiplier to this constraint makes those choices independent
within a fixed witness and produces truncated weights.

\subsection{Truncation and duality}

For $\lambda\geq0$, define the truncated weights $w_\lambda(e)=\min\{w(e),\lambda c(e)\}$ and the corresponding
minimum value $L(\lambda)=\min_{S\in\mathcal F}w_\lambda(S)$.
For a fixed set $S$, retaining an element contributes $w(e)$, whereas
deleting it contributes $\lambda c(e)$ to the Lagrangian objective.
The decisions are independent once the budget constraint is dualized, so
\[
    w_\lambda(S)
    =\sum_{e\in S}\min\{w(e),\lambda c(e)\}
    =\min_{R\subseteq S}\bigl(w(S\setminus R)+\lambda c(R)\bigr).
\]
Hence dualizing the budget constraint in \eqref{eq:interdiction} yields
\begin{equation}\label{eq:dual}
    \Phi(\lambda)
    =\min_{\substack{S\in\mathcal F\\R\subseteq S}}
      \bigl(w(S\setminus R)+\lambda(c(R)-b)\bigr)
    =L(\lambda)-\lambda b.
\end{equation}
Each pair $(S,R)$ defines an affine function of $\lambda$ with slope
$c(R)-b$.  The function $\Phi$ is their lower envelope, and is therefore
concave and piecewise linear.  Weak duality gives
$\Phi(\lambda)\leq\opt$ for every $\lambda\geq0$, since
$\lambda(c(R)-b)\leq0$ for each feasible deletion set.  Write
\[
    \Lambda=\max_{\lambda\geq0}\Phi(\lambda),
    \quad \lambda^*\in\argmax_{\lambda\geq0}\Phi(\lambda),
    \quad L^*=L(\lambda^*).
\]
A finite maximizer exists: $L(\lambda)$ is constant once
$\lambda$ exceeds all ratios $w(e)/c(e)$.
We assume $\opt>0$, or equivalently $\min_{S\in\mathcal F}c(S)>b$.
Otherwise, deleting an entire minimum-cost feasible set gives value $0$.
Under this assumption, $\Lambda>0$.  Indeed, for sufficiently small
$\lambda>0$ every element has truncated weight $\lambda c(e)$, and hence
\[
    \Phi(\lambda)
    =\lambda\left(\min_{S\in\mathcal F}c(S)-b\right)>0.
\]

\subsection{The witness bound}

A pair $(S,R)$ is \emph{active} at $\lambda^*$ if it attains the minimum
in \eqref{eq:dual}.  The following observation bounds $\opt$ from above
by $L^*$.

\begin{lemma}\label{lem:feasible-active}
There is an active pair $(S^{LD},R^{LD})$ with $c(R^{LD})\leq b$.
Consequently, $\opt\leq L^*$.
\end{lemma}
\begin{proof}
If every active pair had $c(R)>b$, every active affine function in
\eqref{eq:dual} would have positive slope.  Since there are finitely many
such functions, $\Phi$ would increase immediately to the right of
$\lambda^*$, contradicting maximality.  Thus a feasible active pair exists.
Its activity and feasibility imply, respectively,
\[
    L^*=w(S^{LD}\setminus R^{LD})+\lambda^*c(R^{LD})
    \geq w(S^{LD}\setminus R^{LD})\geq\opt.
\]
\end{proof}

The lemma bounds $\opt$ from above by $L^*$ even though the dual value
$\Lambda=L^*-\lambda^*b$ is a lower bound.  This relation is the only
optimality property of the multiplier needed in the witness argument.

\begin{theorem}\label{thm:two-approx-witness}
Every optimal solution $(S^*,R^*)$ to \eqref{eq:interdiction} satisfies
\[
    L^*\leq w_{\lambda^*}(S^*)\leq2L^*-\Lambda<2L^*.
\]
\end{theorem}
\begin{proof}
The lower bound follows from $S^*\in\mathcal F$ and the definition of $L^*$.
For the upper bound, use the particular deletion set $R^*$ when evaluating
the truncated weight of $S^*$:
\[
\begin{aligned}
    w_{\lambda^*}(S^*)
    &\leq w(S^*\setminus R^*)+\lambda^*c(R^*)\\
    &\leq\opt+\lambda^*b\\
    &\leq L^*+\lambda^*b
     =2L^*-\Lambda.
\end{aligned}
\]
The second inequality uses $c(R^*)\leq b$, and the third uses
\autoref{lem:feasible-active}.  The strict bound follows from $\Lambda>0$.
\end{proof}

\begin{remark}
The proofs of \autoref{lem:feasible-active} and
\autoref{thm:two-approx-witness}
use only the minimization over pairs $(S,R)$ and the nonnegativity of the
remaining objective.  They therefore extend to a nonnegative set function
$f:2^E\to\mathbb R_{\geq0}$ by defining
\[
    f_\lambda(S)=\min_{R\subseteq S}
        \bigl(f(S\setminus R)+\lambda c(R)\bigr),
    \qquad L_f(\lambda)=\min_{S\in\mathcal F}f_\lambda(S),
\]
and assuming $\max_{\lambda\geq0}(L_f(\lambda)-\lambda b)>0$.
Linearity is used to express $f_\lambda$ as an elementwise truncation and
to reduce the deletion problem on a fixed witness to minimum knapsack.
\end{remark}

\section{The algorithmic reduction}\label{sec:algorithms}

The witness bound gives a direct meta-algorithm: compute $\lambda^*$,
enumerate strict $2$-approximate minimizers under $w_{\lambda^*}$, and
solve the deletion problem on each witness.  We state its running time
in terms of these three subroutines.

\subsection{Enumeration and minimum knapsack}

For a fixed witness $S$, the deletion problem is minimum knapsack:
\begin{equation}\label{eq:knapsack}
    g_b(S)
    =\min_{\substack{R\subseteq S\\c(R)\leq b}}w(S\setminus R)
    =\min_{\substack{Q\subseteq S\\c(Q)\geq c(S)-b}}w(Q).
\end{equation}
Let $T_{\mathrm K}(m,\varepsilon)$ be the time to obtain a feasible
$(1+\varepsilon)$-approximation on at most $m$ elements.  We use
$T_{\mathrm K}(m,\varepsilon)=O(m^2/\varepsilon)$,
as in~\cite[Proposition~7]{huang_connectivity_2026}.
\autoref{fig:alg} gives the resulting meta-algorithm.

\begin{figure}[htbp]
\centering
\begin{algorithm}
\underline{\textsc{Linear-Minimization-Interdiction}$(E,\mathcal F,w,c,b,\varepsilon)$}:\+\\
compute a maximizer $\lambda^*$ of $\Phi(\lambda)=L(\lambda)-\lambda b$\\
compute $w_{\lambda^*}$ and a minimizer $S_0$ of weight $L^*$\\
set $(\hat S,\hat R)\gets(S_0,\emptyset)$\\
enumerate a list containing every $S\in\mathcal F$ with $w_{\lambda^*}(S)<2L^*$\\
for each enumerated $S$:\+\\
    compute a feasible $(1+\varepsilon)$-approximation $Q_S$ to \eqref{eq:knapsack}\\
    set $R_S\gets S\setminus Q_S$\\
    if $w(S\setminus R_S)<w(\hat S\setminus\hat R)$:\+\\
        set $(\hat S,\hat R)\gets(S,R_S)$\-\-\\
return $(\hat S,\hat R)$
\end{algorithm}
\caption{Meta-algorithm for an instance with positive interdiction optimum.}
\label{fig:alg}
\end{figure}

\begin{theorem}\label{thm:reduction}
Suppose $\opt>0$.  Let $T_\lambda$ be the time to compute $\lambda^*$,
$w_{\lambda^*}$, and a minimizer of $w_{\lambda^*}$ over $\mathcal F$.
Suppose a list of at most $K$ members of $\mathcal F$, including every
$S$ with $w_{\lambda^*}(S)<2L^*$, can be produced in time
$T_{\mathrm{enum}}$.  This time includes recovering each witness as a
set of original elements, and $K$ counts repetitions.
Then \autoref{fig:alg} computes a feasible $(1+\varepsilon)$-approximate
interdiction solution in time
\[
    O\bigl(T_\lambda+T_{\mathrm{enum}}
             +K T_{\mathrm K}(m,\varepsilon)\bigr).
\]
If the multiplier and enumeration routines succeed with high probability,
the same is true of the approximation guarantee.
\end{theorem}
\begin{proof}
By \autoref{thm:two-approx-witness}, the list contains an optimal witness
$S^*$.  Minimum knapsack on this witness returns a feasible deletion set
$R_{S^*}$ satisfying
\[
    w(S^*\setminus R_{S^*})
    \leq(1+\varepsilon)g_b(S^*)=(1+\varepsilon)\opt.
\]
The algorithm returns the best feasible candidate, so it has the same
guarantee.  Summing the subroutine costs proves the running-time bound.
\end{proof}

\subsection{Computing the multiplier}\label{sec:multiplier}

We seek an exact maximizer of the concave, piecewise-linear function
$\Phi(\lambda)=L(\lambda)-\lambda b$.  We use Megiddo's parametric
search~\cite{megiddo_rational_1979}; an accessible account, including
concave maximization, is given by Agarwal and
Sharir~\cite[Section~2.1]{agarwal_sharir_geometric_1998}.
The black-box guarantee we need is the following: if a routine determines
in $O(Q)$ time whether a queried parameter is optimal or all maximizers
lie to its left or right, and along each execution path its comparisons
are between affine functions whose coefficients can be maintained in
constant time per operation, then a maximizer can be found in $O(Q^2)$
time.

\begin{lemma}[Computing the multiplier]\label{lem:para}
Suppose linear minimization over $\mathcal F$ takes $Q\geq m$ time,
also when ties are broken by minimizing or maximizing a second additive
weight.  Suppose the algorithm uses the weights only through additions,
subtractions, comparisons, and multiplication or division by
weight-independent quantities.  Then $\lambda^*$, $w_{\lambda^*}$, and
a minimizer under $w_{\lambda^*}$ can be computed in $T_\lambda=O(Q^2)$
arithmetic operations and comparisons.
\end{lemma}
\begin{proof}
At a query $t>0$, let $d_{\min}(t)$ and $d_{\max}(t)$ be the minimum
and maximum $c(R)$ among pairs attaining \eqref{eq:dual}.
Two minimizations under $w_t$ compute these in $O(Q)$ time: break ties
by minimizing $\sum_{e\in S:\,tc(e)<w(e)}c(e)$ and maximizing
$\sum_{e\in S:\,tc(e)\leq w(e)}c(e)$, respectively.
Since $\Phi'_+(t)=d_{\min}(t)-b$ and
$\Phi'_-(t)=d_{\max}(t)-b$, concavity determines whether $t$ is optimal
or which side contains all maximizers.

Computing $w_t$ and the secondary weights compares $tc(e)$ with $w(e)$.
Along each execution path, the resulting weights are affine in $t$,
and the assumed operations preserve this property.  Thus the direction
test satisfies the black-box requirements, giving $O(Q^2)$ time,
including a final minimization at $\lambda^*$.
\end{proof}

\section{Applications}\label{sec:applications}

The applications below use witness families defined in the original
instance.  For graphs, a vertex partition determines
the set of crossing edges.  For matroids, cocircuits
describe the minimal sets whose removal lowers the original rank.
Deleting elements removes exactly their contributions to the chosen
witness, so the interdiction objective is \eqref{eq:interdiction}.
The differences between the applications are the optimization and
enumeration routines required by \autoref{thm:reduction}.

\subsection{Graph connectivity interdiction}\label{sec:connectivity}

Here $\mathcal F$ is the family of graph cuts.  A minimum cut under
costs $c$ detects whether $\opt=0$; otherwise, evaluating $L(\lambda)$
is a minimum-cut computation under weights $w_\lambda$.

\begin{theorem}\label{thm:connectivity}
Connectivity interdiction admits a deterministic FPTAS with running time
$\tilde O(n^3m^2/\varepsilon)$.
\end{theorem}
\begin{proof}
Stoer and Wagner~\cite{stoer_wagner_1997} give a deterministic
minimum-cut algorithm taking $Q_G=O(nm+n^2\log n)$ time.  Its additions
and comparisons satisfy the arithmetic and tie-breaking requirements of
\autoref{lem:para}, giving $T_\lambda=O(Q_G^2)$.

Xu and Yang~\cite[Theorem~2]{xu_yang_small_cuts_2026} show that a
connected undirected graph with nonnegative real edge weights and positive
minimum cut value $L$ has $O(n^3)$ cuts of weight strictly below $2L$.
Since $\opt>0$, we have $L^*>0$, and applying their theorem under
$w_{\lambda^*}$ gives $K=O(n^3)$ candidate cuts.  By
\autoref{thm:two-approx-witness}, these include every optimal witness.

To enumerate these cuts, use the deterministic algorithm of Yeh, Wang,
and Su~\cite{yeh_enumerating_cuts_2010}, which
lists cuts without repetition in nondecreasing order of weight, with
$\tilde O(nm)$ delay per cut.  For this enumeration,
merge parallel edges by summing their truncated weights, while keeping
the original edges for the subsequent knapsack problems.  Stop at the
first cut of weight at least $2L^*$, or when no cut remains.
Recovering the original edge set of each retained cut takes $O(m)$ time.
Thus
\[
    K=O(n^3),\qquad
    T_{\mathrm{enum}}=\tilde O\bigl((K+1)nm+Km\bigr)
                    =\tilde O(n^4m).
\]
Substituting these bounds into \autoref{thm:reduction} gives
\[
    \tilde O\bigl(Q_G^2+n^4m+n^3T_{\mathrm K}(m,\varepsilon)\bigr)
    =\tilde O(n^3m^2/\varepsilon),
\]
using $T_{\mathrm K}(m,\varepsilon)=O(m^2/\varepsilon)$ and
$m\geq n-1$ in a positive-optimum instance.  All subroutines are deterministic.
\end{proof}

\paragraph{$k$-cuts.}\label{sec:kcut}
For $k\geq3$, take $\mathcal F$ to be the crossing-edge sets
of partitions of $V$ into $k$ nonempty parts.  Optimization under
$w_\lambda$ and the zero-optimum test are minimum $k$-cut computations.

\begin{theorem}\label{thm:kcut}
For $k\geq3$, $k$-cut interdiction admits an FPRAS in
time $\tilde O_k(n^{2k}m^2/\varepsilon)$.
\end{theorem}
\begin{proof}
Gupta, Harris, Lee, and Li~\cite{gupta_kcut_2022} give
minimum $k$-cut in $Q_k=\tilde O_k(m+n^k)$ time and enumerate all
$2$-approximate minimum $k$-cuts using $K=\tilde O_k(n^{2k})$
compressed candidates, with high probability.
With fixed random choices, their weighted sampling and cut comparisons
use affine expressions along each execution path.  Repeating the randomized
minimization routine with fresh randomness makes all direction tests
correct with high probability; the repetition cost is included in $Q_k$.
The construction of \autoref{lem:para} therefore gives
$T_\lambda=O(Q_k^2)$ with high probability.
Recovering the original edges of each candidate takes $O(n+m)$ time,
giving $T_{\mathrm{enum}}=\tilde O_k(n^{2k}(n+m))$.
By \autoref{thm:reduction}, the running time is
\[
    \tilde O_k\bigl((m+n^k)^2
              +n^{2k}(n+m+T_{\mathrm K}(m,\varepsilon))\bigr)
    =\tilde O_k(n^{2k}m^2/\varepsilon).
\]
Here $n=O_k(m)$ in every positive-optimum instance, since otherwise
$G$ already has at least $k$ connected components.
\end{proof}

\subsection{Cogirth interdiction on perturbed graphic matroids}
\label{sec:matroid-interdiction}

A rank-$p$ perturbed graphic matroid $M$ on $m$ elements is represented
over $\mathbb F_2$ by $A(G)+P$, where $A(G)$ is a graph incidence
matrix and $\rank(P)\leq p$.  We assume that $M$ is loopless and has
positive rank, and write $\mathcal C^*(M)$ for its cocircuit family.
For positive element weights $a$, write
$\kappa(M,w)=\min_{D\in\mathcal C^*(M)}w(D)$ for its cogirth.
We measure vulnerability relative to the original rank: after deleting
$R$, the remaining vulnerability is
\[
    \kappa_R(M,w)=\min_{D\in\mathcal C^*(M)}w(D\setminus R).
\]
It is zero if $R$ has already lowered the rank of $M$; otherwise it
equals the cogirth of $M\setminus R$.  Thus minimizing $\kappa_R(M,w)$
subject to $c(R)\leq b$ is exactly \eqref{eq:interdiction} with
$\mathcal F=\mathcal C^*(M)$.

For enumeration, we use the following consequence of \cite[Section 5.2]{cong_contraction_2026}.

\begin{lemma}\label{lem:pgm-enumeration}
For a loopless rank-$p$ perturbed graphic matroid $M$ of positive rank
on $m$ elements with positive weights $a$, a list containing every cocircuit
of weight at most $2\kappa(M,a)$ can be computed deterministically
in $2^{O(p^2)}m^{O(1)}$ time.
\end{lemma}

\begin{corollary}\label{cor:pgm-interdiction}
Cogirth interdiction on a rank-$p$ perturbed graphic matroid, given by
its representation $A(G)+P$, admits a deterministic FPTAS in time
\[
    2^{O(p^2)}m^{O(1)}/\varepsilon.
\]
\end{corollary}
\begin{proof}
The algorithm of~\cite{cong_contraction_2026} gives the minimum-cocircuit
oracle in $Q_p=2^{O(p^2)}m^{O(1)}$ time.  Applying it under costs $c$
first detects optimum zero.  Its packing steps can use the combinatorial
matroid-polyhedron routines of
Cunningham~\cite{cunningham_membership_1984}; the weight operations are
additions, subtractions, comparisons, and scaling by weight-independent
rank quantities.  Thus the oracle satisfies the arithmetic and
tie-breaking requirements of \autoref{lem:para}, which gives
$T_\lambda=O(Q_p^2)$.
By \autoref{lem:pgm-enumeration}, both $T_{\mathrm{enum}}$ and $K$ are
$2^{O(p^2)}m^{O(1)}$.  Substituting into \autoref{thm:reduction}, with
$T_{\mathrm K}(m,\varepsilon)=O(m^2/\varepsilon)$, gives
\[
    T_\lambda+T_{\mathrm{enum}}+K T_{\mathrm K}(m,\varepsilon)
    =2^{O(p^2)}m^{O(1)}/\varepsilon.
\]
All subroutines are deterministic.  Enumeration reports original
elements, so minimum knapsack retains their individual weights and
costs even when the oracles collapse parallel classes internally.
\end{proof}

\bibliographystyle{plain}
\bibliography{ref}

\end{document}